\documentclass[runningheads]{llncs}

\usepackage[utf8]{inputenc} 
\usepackage[T1]{fontenc}    
\usepackage{hyperref}       
\usepackage{url}            
\usepackage{booktabs}       
\usepackage{amsfonts}       
\usepackage{nicefrac}       
\usepackage{microtype}      

\usepackage{amssymb}
\usepackage{amsmath}
\usepackage{url}
\usepackage{color}
\usepackage{graphicx}
\usepackage{balance}  
\usepackage{stmaryrd} 
\usepackage{thm-restate} 
\usepackage[linesnumbered,lined,ruled,noend]{algorithm2e}
\usepackage[T1]{fontenc}
\usepackage{lmodern}

\newcommand{\each}{\emph{\bf each} }
\newcommand{\be}{\begin{equation}}
\newcommand{\ee}{\end{equation}}

\newcommand{\ZZ}{{\mathbb Z}}

\newcommand{\NN}{{\mathbb N}}
\newcommand{\PP}{{\mathbb P}}

\newcommand{\Gaussian}{\mathcal{N}}

\newcommand{\dist}{\mathrm{dist}}

\newcommand{\norm}[1]{\left\lVert #1\right\rVert}                               %

\newcommand{\br}[1]{\left\{#1\right\}}

\newcommand{\abs}[1]        {\left| #1\right|}

\newcommand{\REAL}{\ensuremath{\mathbb{R}}}                       
\newcommand{\RR}{{\REAL}}

\newcommand\sett[2]{\left\{ \left. #1 \;\right\vert #2 \right\}}

\newcommand\details[1]{}

\newcommand{\assign}{:=}

\newcommand{\floor}[1]{\lfloor #1\rfloor}
\newcommand{\ceil}[1]{\lceil #1\rceil}
\newcommand{\round}[1]{\lceil #1 \rfloor}
\newcommand{\encrypted}[1]{\llbracket #1\rrbracket}
\newcommand{\enc}[1]{\llbracket #1\rrbracket}

\newcommand{\Sim}{\mathcal{S}}
\newcommand{\A}{\mathcal{A}}

\newcommand{\E}{\mathcal{E}}
\newcommand{\Gen}{\operatorname{Gen}}
\newcommand{\Enc}{\operatorname{Enc}}

\newcommand{\view}{\textsf{view}}

\newcommand{\SD}{\mathrm{SD}}
\newcommand{\DD}{\mathcal{D}}
\newcommand{\T}{\mathcal{T}}

\newcommand{\isSmaller}{\mathrm{isSmaller}}
\newcommand{\isNeg}{\mathrm{isNegative}}

\newcommand{\size}{\mathrm{size}}
\newcommand{\depth}{\mathrm{depth}}
\newcommand{\overhead}{{overhead}}
\newcommand{\indegree}{\mathrm{indegree}}

\newcommand{\val}{\mathrm{val}}

\newcommand{\class}{\mathrm{classes}}

\newcommand{\CoinToss}{\mathrm{CoinToss}}

\newcommand{\KNearestNeighbors}{\mathrm{KNearestNeighbors}}

\newcommand{\ProbAvg}{\mathrm{ProbabilisticAverage}}
\newcommand{\computeDist}{\mathrm{computeDist}}

\newcommand{\high}{\mathrm{high}}

\newcommand{\low}{\mathrm{low}}

\newenvironment{protocol}[1][htb]
  {
   \begin{algorithm}[#1]%
  }{\end{algorithm}}

\begin{document}

\title{Secure $k$-ish Nearest Neighbors Classifier}

\author{
Hayim Shaul\inst{1} \and
Dan Feldman\inst{2} \and
Daniela Rus\inst{1}
}

\institute{CSAIL MIT, Cambridge, MA, USA.
\email{\{hayim,rus\}@csail.mit.edu}\and
University of Haifa, Haifa, Israel.
\email{\{dannyf\}@csail.mit.edu}
}


\maketitle

\begin{abstract}
In machine learning, classifiers are used to predict a class of a given query based on an existing (already classified) database.
Given a database $S$ of $n$ $d$-dimensional points and a $d$-dimensional query $q$,
the $k$-nearest neighbors ($k$NN) classifier assigns $q$ with the majority class of its $k$ nearest neighbors in $S$.

In the secure version of $k$NN, $S$ and $q$ are owned by two different parties that do not want to share their (sensitive) data.
A secure classifier has many applications. For example, diagnosing a tumor as malignant or benign.
Unfortunately, all known solutions for secure $k$NN either require a large communication complexity between the parties,
or are very inefficient to run (estimated running time of weeks).

In this work we present a classifier based on $k$NN, that can be implemented efficiently with homomorphic encryption (HE).
The efficiency of our classifier comes from a relaxation we make on $k$NN, where we allow it to consider $\kappa$ nearest neighbors for
$\kappa\approx k$ with some probability. We therefore call our classifier $k$-ish Nearest Neighbors ($k$-ish NN).

The success probability of our solution depends on the distribution of the distances from $q$ to $S$ and increase as its
statistical distance to Gaussian decrease.

To efficiently implement our classifier we introduce the concept of {\em doubly-blinded coin-toss}. In a doubly-blinded coin-toss
the success probability as well as the output of the toss are encrypted.
We use this coin-toss to efficiently approximate the average and variance of the distances from $q$ to $S$.
We believe these two techniques may be of independent interest.

When implemented with HE, the $k$-ish NN has a circuit depth (or polynomial degree) that is independent of $n$, therefore making it scalable.
We also implemented our classifier in an open source library based on HELib~\cite{HELib} and tested it on a breast tumor database.
The accuracy of our classifier (measured as $F_1$ score) were $98\%$ and classification took less than 3 hours compared to (estimated) weeks in current HE implementations.

\end{abstract}

\section{Introduction}
\label{sec_intro}

A key task in machine learning is to classify an object based on a database of previously classified objects.
For example, with a database of tumors, each of them  classified as malignant or benign,
we wish to classify a new tumor.
Classification algorithms have been long studied. For example, $k$ nearest neighbor ($k$NN) classifier~\cite{knn1967}, where a classification of a new tumor is
done by considering the $k$ nearest neighbors (i.e. the most similar tumors, for some notion of similarity).
The decision is then taken to be the majority class of those neighbors.

In some cases, we wish to perform the classification without sharing the database or the query.
In our example, the database may be owned by a hospital while the query is done by a clinic.
Here, sharing the database is prohibited by regulations (e.g. HIPPA~\cite{hipaa})
and sharing the query may expose the hospital and the clinic to liabilities and regulations (e.g. HIPAA~\cite{hipaa} and GDPR~\cite{gdpr}).

In secure multi-party computation (MPC), several parties compute a function without sharing their input.
Solutions such as \cite{mpc-triplets,yao} have the disadvantage of having a large communication complexity. Specifically it is proportional to the running time of
the computation.
Recently, secure-MPC have been proposed based on homomorphic encryption (HE) (see~\cite{BGV12}) that makes it possible to compute a polynomial over encrypted messages (ciphertexts).
Using HE, the communication complexity becomes proportional to the size of the input and output.
In our example, the clinic encrypts its query with HE and sends the encrypted query to the hospital.
The polynomial the hospital applies is evaluated to output a ciphertext that can be decrypted, (only) by the clinic, to get the classification of the query.
See Figure~\ref{fig_system}.

\begin{figure*}[htbp]
\begin{center}
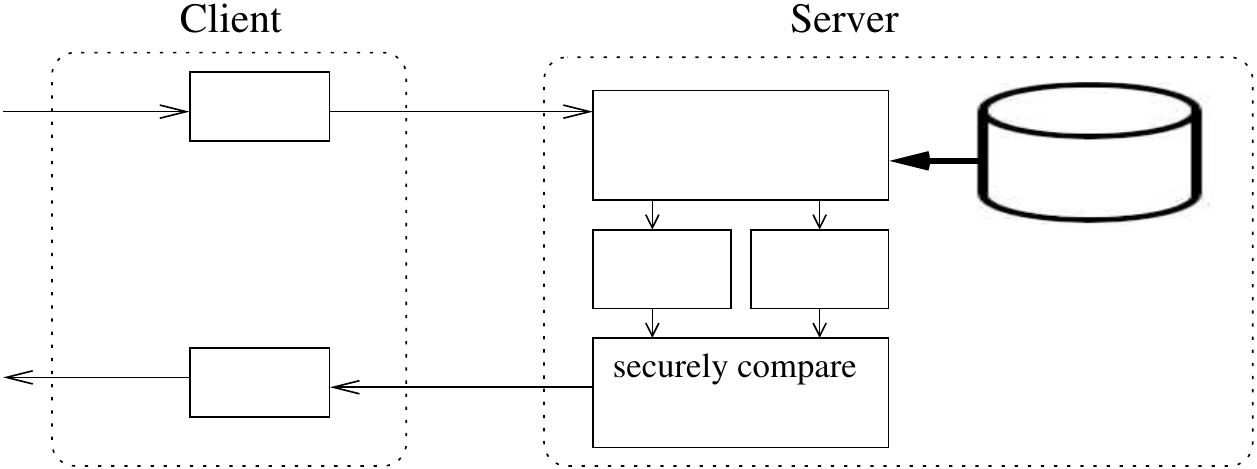
\end{center}
\caption{A HE-based protocol for Secure $k$-nearest neighbors classifier.
(i) A client has a pair $(sk,pk)$ and a query $q$. The client encrypts the query $\encrypted{q} = Enc_{pk}(q)$ and sends $\encrypted{q}$ and $pk$ to the server.
(ii) The Server securely finds the $k$ neighbors of $\encrypted{q}$ from $s_1,\ldots,s_n$.
(iii) The Server securely counts $\encrypted{C_0}$ and $\encrypted{C_1}$, the number of neighbors having class 0 and 1, respectively. Since these are counted
with HE the result, $\encrypted{C_0}$ and $\encrypted{C_1}$, are also encrypted.
(iv) The server computes a polynomial over $\encrypted{C_0}$ and $\encrypted{C_1}$ that determines $\encrypted{class_q}$, the class of $q$.
(iv) The server sends the ciphertext $\encrypted{class_q}$ to the client.
(v) The client decrypts $class_q = Dec_{sk}(\encrypted{class_q})$.
}
\label{fig_system}
\end{figure*}

The downside of using HE is the efficiency of evaluating polynomials.
Although generic recipes exists that formulate any algorithm as a polynomial of its input, in practice the polynomials generated by these recipes
have poor performance.
The main reason for the poor performance is the lack of any comparison operators. Since comparisons  leak information that can be used to break the encryption,
under HE we can only have a ``comparison'' polynomial whose output is encrypted and equals 1 if the comparison holds and 0 otherwise.
The second reason is an artifact of homomorphic encryption schemes: the overhead of evaluating a single operation grows with the degree of the evaluated polynomial.
For many ``interesting'' problems it is a challenge to construct a polynomial that can be efficiently evaluated  with HE.

In this paper we consider the secure classification problem.
We propose a new classifier which we call $k$-ish nearest neighbors. In this new classifier the server considers some $\kappa \approx k$ nearest neighbors to the query.
Relaxing the number of neighbors significantly improves the time performance of our classifier while having a small impact on the accuracy performance.
Specifically, the time to compute our classifier on real breast cancer database dropped from weeks (estimated) to less than 3 hours, while the accuracy (measured by $F_1$ score) decreased from $99\%$ to $98\%$.
See details in Section~\ref{sec_experiment}.

The solution we introduce in this paper assumes
the distances of the database to the query are statistically close to Gaussian distribution.
Although sounding too limiting, we argue (and show empirically) otherwise.
We show that many times the distribution of distances is statistically close enough to Gaussian.
In future work, we intend to remove this assumption.

The efficiency of our solution comes from two
new non-deterministic primitives that we introduce in this paper:
\begin{itemize}
\item a new approach to efficiently compute an approximation to $1/m \sum_{i=1}^n f(\encrypted{x_i})$, where $n,m$ are integers, $f$ is an increasing invertible function and $\encrypted{x_1}, \encrypted{x_2}, \ldots$ are ciphertexts.
\item a {\em doubly-blinded} coin-tossing algorithm, where the result {\bf and} the probability of the toss are encrypted.
\end{itemize}
We believe these two primitives are of independent interest and can be used in other algorithm as well

We built a system written in C++ and using HElib~\cite{HELib} to securely classify breast tumor as benign or malignant using $k$-ish NN classifier.
Our classifier used the Wisconsin Diagnostic Breast Cancer Data Set~\cite{cancerData},
classified a query in less than 3 hours with $98\%$ accuracy.
This significantly improves over previous running times and makes secure classifications with HE a solution that is practical enough to be implemented.

\section{Related Work}
Previous work on secure $k$NN either had infeasible running time or had a large communication complexity.
For example,
in~\cite{Wong_sknn_2009}, Wong et al. considered a distance recoverable encryption
to have the server encrypt $S$. The user encrypts $q$ and a management system computes and compares the distances.
However, this scheme leaks information to an attacker knowing some of the points in $S$~\cite{Xiao_sknn_2013}, in addition some data leaks to the management system.
In~\cite{cheng_sknn_15}, Cheng et al. proposed a similar faster scheme where their  speedup was achieved by using a secure index.
Their scheme relies upon a cloud server that does not collude with the data owner, a model which is not always possible or trusted.
In~\cite{Yiu_sknn_2010}, Yiu et at. considered the two dimensional case where users want to hide their locations.
They use a cryptographic transformation scheme to let the user perform a search on an R-tree, which requires as many protocol rounds as the height of the tree.
In addition, most of their computation is done in by the user which might not always be feasible.
In~\cite{Hu_sknn_2011}, Hu et al. propose a scheme to traverse an R-tree where a homomorphic encryption scheme is used to compute distances and choose the next node in the traversing
of the R-tree.
However, this scheme is vulnerable if the attacker knows some of the points in $S$~\cite{Xiao_sknn_2013}.
In addition, the communication complexity is proportional to the height of the R-tree.
In~\cite{Zhu_sknn_2013}, Zhu et al. considered a security model in which users are untrusted. In their scheme the Server communicates with the cloud to compute the set of 
$k$ nearest neighbors to the query $q$, however, this leaks some information on $q$.
In~\cite{elmehdwi_sknn_14}, Elmehdwi et al. proposed a scheme that is, to the best of our knowledge, the first to guarantee privacy of data as well as query.
However, this scheme depends on the existence of a non colluding cloud. This is an assumption not all users are willing to make. In addition, the communication overhead of
this scheme is very high (proportional to the size of the database).

\section{Preliminaries}
For an integer $m$ we denote $[m] = \br{1,\ldots, m}$.
We use $\encrypted{msg}$ to denote a ciphertext that decrypts to the value $msg$.

The field $\ZZ_p$, where $p$ is prime, is the set $\br{0,\ldots, p-1}$ equipped with $+$ and $\cdot$ done modulo $p$.

We denote by $\round{x}$, where $x \in \RR$, the rounding of $x$ to the nearest integer.

A {\em database} of $\ZZ_p^d$ points of size $n$ is the tuple $S=(s_1, \ldots, s_n)$, where $s_1, \ldots, s_n \in \ZZ_p^d$.
We denote by $class(s_i) \in \br{0,1}$ the class of $s_i$.

Let $S=(s_1, \ldots, s_n)$ be a database of $\ZZ_p^d$ points of size $n$ and let $q\in\ZZ_p^d$.
The {\em distance distribution} is the distribution of the random variable $x = dist(s_i,q)$, where $i \leftarrow [n]$ is drawn uniformly.
We denote the distance distribution by $\DD_{S,q}$.

The statistical distance between two discrete probability distribution $X$ and $Y$ over a finite set $\br{0,\ldots, p-1}$, denoted $\SD(X, Y)$, is defined
as
$$ SD(X,Y) = \max_{u\in\br{0,\ldots,p-1}} \abs{ Pr[x=u] - Pr[y=u] },$$ where $x \sim X$ and $y \sim Y.$

The cumulative distribution function (CDF) of a distribution $X$ is defined as $CDF_X(\alpha) = Pr[ x < \alpha \mid x \sim X]$.

The {\em $F_1$ Score} (also called {\em Dice coefficient} or {\em Sorensen} coefficient) is a measure of similarity of two sets. It is given by
$F_1 = 2\frac{|X \cap Y|}{|X| + |Y|}$, where $X$ and $Y$ are sets. In the context of classifier, the $F_1$ score is used to measure the accuracy of
a classifier by taking $X$ to be the set of samples classified as 1 and $Y$ be the set of samples whose class is 1.

\subsection{Polynomial Interpolation}
For a prime $p$ and a function $f : [0,p] \mapsto [0,p]$, we define the polynomial $\PP_{f,p} : \ZZ_p \mapsto \ZZ_p$, where $\PP_{f,p}(x) = \round{f(x)}$ for
all $x \in \ZZ_p$. When $p$ is known from the context we simple write $\PP_f$.

An explicit description of $\PP_{f,p}$ can be given by the interpolation
$$\PP_{f,p} (x) = \sum_{i=0}^{p-1} \frac{(x-0)\cdots\big(x-(i-1)\big)\cdot\big(x-(i+1)\big)\cdots\big(x-(p-1)\big)}{(i-0)\cdots\big(i-(i-1)\big)\cdot\big(i-(i+1)\big)\cdots\big(i-(p-1)\big)}\round{ f(i) }.$$
Rearranging the above we can write $$\PP_{f,p}(x) = \sum_{i=0}^{p-1} \alpha_i x^i$$ for appropriate coefficients $\alpha_0,\ldots,\alpha_{p-1}$ that depend on $f$.
In this paper we use several polynomial interpolations:

\begin{itemize}
\item $\PP_{\sqrt{\cdot}}(x) = \round{\sqrt{x}}$.
\item $\PP_{(\cdot)^2/p}(x) = x^2/p$.
\item $\PP_{(\cdot=0)}(x) = 1$ if $x=0$ and 0 otherwise.
\item $\PP_{(\sqrt{(\cdot)+p}}(x) = \sqrt{x + p}$.
\item $\PP_{(\sqrt{(\cdot)p}}(x) = \sqrt{xp}$.
\item $\PP_{\isNeg(\cdot)}(x) = 1$ if $x > p/2$ and 0 otherwise.
\end{itemize}

\paragraph{Comparing Two Ciphertexts.}
We also define a two variate function $\isSmaller : \ZZ_{p}\times\ZZ_{p} \mapsto \br{0,1}$, where $\isSmaller(x,y) = 1$ iff $x<y$.
In this paper we implement this two-variate function with a uni-variate polynomial $\isNeg : \ZZ_{p'} \mapsto \br{0,1}$,
where $p'>2p$ and $\isNeg(z) = 1$ iff $z > p'/2$.
The connection between these two polynomials is given by $$\isSmaller_p(x,y) = \isNeg_{p'}(x-y).$$

\paragraph{Computing Distances.}
\label{sec_compute_dist}
Looking ahead,
our protocol can work with any implementation of a distance function. In this paper we analyzed and tested our protocol with the $\ell_1$ distance.
We implemented a polynomial $\dist_{\ell_1}(a,b)$ that evaluates to $\norm{a-b}_{\ell_1}$ where $a=(a_1, \ldots, a_d)$ and $b= (b_1, \ldots, b_d)$:
$$\dist_{\ell_1}(a,b) = \sum_1^d \big(1-2\isSmaller(a_i, b_i)\big)(a_i - b_i).$$

We observe that
\[
\big(1-2\cdot\isSmaller(a_i, b_i)\big) =
\begin{cases}
1 & \text{if } a_i<b_i\\
-1 & \text{otherwise}.
\end{cases}
\]
and therefore $\sum_i \big(1-2\cdot\isSmaller(a_i,b_i)\big)(a_i - b_i) = \sum_i \abs{a_i - b_i} = \dist_{\ell_1}(a,b)$.

%
%
%

\paragraph{Arithmetic Circuit vs. Polynomials.}
An arithmetic circuit (AC) is a directed graph $G=(V,E)$ where for each node $v\in V$ we have $\indegree(v) \in \br{0,2}$,
where $\indegree(v)$ is the number of incoming edges of $v$.
We also associate with each $v\in V$ a value, $\val(v)$ in the following manner:

If $\indegree(v) = 0$ then we call $v$ an {\em input node} and
we associate it with a constant or with an input variable and $\val(v)$ is set to that constant or variable.

If $\indegree(v) = 2$ then we call $v$ a {\em gate} and
associate $v$ with an $Add$ operation (an {\em $add$-gate}) or $Mul$ operation (a {\em $mult$-gate}).

Denote by $v_1$ and $v_2$ the nodes connected to $v$ through the incoming edges and
set $\val(v) \assign \val(v_1) + \val(v_2)$  if $v$ is an $add$-gate or $\val(v) \assign \val(v_1) \cdot \val(v_2)$ if $v$ is a $mult$-gate.

Arithmetic circuits and polynomials are closely related but do not have a one-to-one correspondence. A polynomial can be realized by different circuits.
For example $\PP(x)=(x+1)^2$ can be realized as $(x+1)\cdot(x+1)$ or by $x\cdot x + 2\cdot x + 1$.
The first version has one $mult$-gate and one $add$-gate (here we take advantage that $x+1$ needs to be calculated once), while the latter has 2 $mult$-gates and 2 $add$-gates.

Looking ahead, we are going to evaluate arithmetic circuits whose gates are associated with HE operations (see below).
We are therefore interested in bounding two parameters of an arithmetic circuits, $C$:
\begin{itemize}
\item $\size(C)$ is the number of $mult$ gates in $C$.
This relates to the number of gates needed to be evaluated hence directly affecting the running time.
(We discuss below why we consider only $mult$ gates.
$Mul$ gates).
\item $\depth(C)$ is the maximal number of $mult$ gates in a path in $C$.
(We discuss below why we consider only
\end{itemize}

Given a function $f:[0,p]\mapsto[0,p]$  Paterson et al. showed in \cite{sqrt-mult-73} that a polynomial $\PP_{f,p}(x)$ can be realized by an arithmetic
circuit $C$ where  $\depth(C) = O(\log p)$ and $\size(C) = O(\sqrt{p})$.
Their construction can be extended to realize multivariate polynomials, however, size of the resulting circuit grows exponentially with the number of
variables.
For a $n$-variate polynomial, such as a polynomial that evaluates to the $k$ nearest neighbors, their construction has poor performance.

\subsection{Homomorphic Encryption}
Homomorphic encryption (HE) \cite{BGV12,GentrySTOC09} (see also a survey in\cite{HaleviShoup2014helib})
is an asymmetric encryption scheme that also supports + and $\times$ operations on ciphertexts.
More specifically, HE scheme is the tuple $\E = (Gen, Enc, Dec, Add, Mul)$, where:
\begin{itemize}
	\item $Gen(1^\lambda, p)$ gets a security parameter $\lambda$ and an integer $p$ and generates a public key $pk$ and a secret key $sk$.
	\item $Enc_{pk}(m)$ takes a message $m$ and outputs a ciphertext $\encrypted{m}$.
	\item $Dec_{sk}(\encrypted{m})$ gets a ciphertext $\encrypted{m}$ and outputs a message $m'$.

{\em Correctness} is the standard requirement that $m' = m$.

	\item $Add_{pk}(\encrypted{a}, \encrypted{b})$ gets two ciphertexts $\encrypted{a},\encrypted{b}$ and outputs
			a ciphertext $\encrypted{c}$.

{\em Correctness} is the requirement that $c = a + b \mod p$.

	\item $Mul_{pk}(\encrypted{a}, \encrypted{b})$ gets two ciphertexts $\encrypted{a},\encrypted{b}$ and outputs
			a ciphertext $c$.

{\em Correctness} is the requirement that $c = a \cdot b \mod p$.

\end{itemize}

\paragraph{Abbreviated syntax.}

To make our algorithms more intuitive we use $\encrypted{\cdot}_{pk}$ to denote a ciphertext.
When $pk$ is clear from the context we use an abbreviated syntax:

\begin{itemize}
	\item $\encrypted{a} + \encrypted{b}$ is short for $Add_{pk}(\encrypted{a}, \encrypted{b})$.
	\item $\encrypted{a} \cdot \encrypted{b}$ is short for $Mul_{pk}(\encrypted{a}, \encrypted{b})$.

	\item $\encrypted{a} + b$ is short for $Add_{pk}(\encrypted{a}, Enc_{pk}(b))$.
	\item $\encrypted{a} \cdot b$ is short for $Mul_{pk}(\encrypted{a}, Enc_{pk}(b))$.
\end{itemize}

Given these operations any polynomial $\PP(x_1, \ldots)$ can be realized as an arithmetic circuit and computed on the ciphertexts $\encrypted{x_1},\ldots$.
For example, in a client-server computation the client encrypts its data and sends it
to the server to be processed.
The polynomial the server evaluates is dependent on its input and the result is a ciphertext that is returned to the client.
The client then decrypts the output.
The semantic security of homomorphic encryption guaranties the server does not learn anything on the client's data. Similarly, the client does not learn
anything on the server except for the output.

The cost of evaluating an arithmetic circuit, $C$, with HE is $$Time = \overhead \cdot \size(C),$$ where $\overhead$ is the time to evaluate one $mult$-gate and
it varies with the underlying implementation of the
HE scheme. For example, for BGV scheme~\cite{BGV12} we have $\overhead = \overhead(C) = O\big(\big(\depth(C)\big)^3\big)$.

\paragraph{Threat Model.}
The input of the server in our protocol is $n$ points, $s_1, \ldots, s_n$ with their respective classes.
The input of the client is a query point $q$.
The output of the client is the class of $q$ as calculated by the server based on the classes of the nearest neighbors of $q$.
The server has no output. 
We consider adversaries that are computationally-bounded and semi-honest, i.e. the adversary follows the protocol but may try to learn additional information.
Our security requirement is that the adversary does not learn any information except what is explicitly leaked in the protocol.
Looking ahead, from the semantic security of the HE scheme, the view of the Client is exactly its input and its output.
In addition the view of the server is only its input (the server has no output),
up to negligible information that can feasibly be extracted from HE-encrypted messages.

\section{Our Contributions}
{\bf New ``HE friendly'' classifier.}
In this paper we introduce a new classifier
that we call the {\em $k$-ish nearest neighbor} classifier and is a variation of the $k$ nearest neighbors classifier (see Definition~\ref{def_kish_nn}).


Informally, this classifier considers $\kappa \approx k$ neighbors of a given query.
This relaxation allows us to
implement this classifier with an arithmetic circuit with low depth independent of the database size and has significantly better running times than $k$NN.

The success probability of the implementation we give in this paper depends on $\SD\big(\DD_{S,q}, \Gaussian(\mu,\sigma)\big)$, where $\mu = E(\DD_{S,q})$ and $\sigma^2 = Var(\DD_{S,q})$.
In future papers we intend to propose implementations that do not have this dependency.

{\bf System and experiments for secure $k$-ish NN classifier.}
We implemented our algorithms into a system that uses secure $k$-ish nearest neighbors classifier.
Our code based on HELib~\cite{HELib} and LiPHE~\cite{liphe} is provided for the community
to reproduce our experiments, to extend our results for
real-world applications, and for practitioners at industry or academy that wish to use these results for their
future papers or products.

{\bf A new approximation technique.}
Our low-depth implementation of the $k$-ish NN, is due to a new approximation technique we introduce.
We consider the sum $1/m \sum_1^n f(\encrypted{x_i})$, when $m$ and $n$ are integers, $f$ is an increasing invertible function
and $\encrypted{x_1}, \ldots, \encrypted{x_n}$ are ciphertexts.
We show how this sum can be approximated in a polynomial of degree independent of $n$. Specifically, our implementation
does not require a large ring size to evaluate correctly.



In contrast,
previous techniques for computing similar sums (e.g. for computing average in ~\cite{vinod_avg}), either generate large intermediate values or are realized with a deep arithmetic circuit.

{\bf A novel technique for doubly-blinded coin tossing.}
Our non-deterministic approximation relies on a new algorithm we call {\em doubly-blinded coin toss}.
In a doubly-blinded coin toss the probability of the toss is $\encrypted{x}/m$, where $\encrypted{x}$ is a ciphertext and $n$ is a parameter.
The output of the toss is also  ciphertext.
To the best of our knowledge, this is the first efficient implementation of a coin toss where the probability depends on a ciphertext.
Since coin tossing is a basic primitive for many random algorithms, we expect our implementation of coin tossing to have a large impact on future research in HE.

%
%
%

\section{Techniques Overview}
We give an overview of our techniques that we used.
We first describe the intuition behind the $k$-ish NN and then
we give a set of reductions from the $k$-ish NN classifier to a doubly-blinded coin toss.

\paragraph{Replacing $k$NN with $k$-ish NN.}
Given a database $S$ and a query $q$, small changes in $k$ have small impact on the output of the $k$NN classifier.
We therefore define the $k$-ish NN, where for a specified $k$ the classifier may consider $k/2 < \kappa < 3k/2$.
With this relaxation, our implementation applies non-deterministic algorithms (see below) to find $k$-ish nearest neighbors.

\paragraph{Reducing $k$-ish NN to computing moments.}
For the distance distribution $\DD_{S,q}$, denote $\mu = E(\DD_{S,q})$ and $\sigma^2 = Var(\DD_{S,q})$ and
consider the set $\T = \sett{s_i}{\dist(s_i, q) < T}$, where $T = \mu + \Phi^{-1}(k/n)\sigma$ and $\Phi$ is the CDF of $\Gaussian(0,1)$.
For the case $\DD_{S,q}=\Gaussian(\mu, \sigma)$ we have $|\T|=k$, otherwise since $\dist(x_i, q) > 0$ are from a discrete set
we have $\Big\lvert{\abs{\T} - k}\Big\rvert < \SD\big(\DD_{S,q}, \Gaussian(\mu,\sigma)\big)\cdot T \cdot n$.

For distance distribution that are statistically close to Gaussian it remains to show how $\mu$ and $\sigma$ can be efficiently computed.
We remark that $\mu = \frac{1}{n}\sum \dist(s_i, q)$ and $\sigma = \sqrt{\mu^2 - \frac{1}{n}\sum \big(\dist(s_i, q)\big)^2 }$,
so it remains to show how the first two moments $\frac{1}{n}\sum \dist(s_i, q)$ and $\frac{1}{n}\sum \big(\dist(s_i, q)\big)^2$  can be efficiently computed.

\paragraph{Reducing computing moments to doubly-blinded coin-toss.}
We show how to compute $\frac{1}{n}\sum f(x_i)$ for an increasing invertible function $f$. For an average take $f(x_i) = x_i$ and for the second moment take $f(x_i) = x_i^2$.
Observe that $\frac{1}{n}\sum f(x_i) = \sum \frac{f(x_i)}{n}$ is approximated by
$\sum a_i$, where
\[
a_i =
\begin{cases}
1 & \text{ With probability } f(x_i)/n\\
0 & \text{otherwise}.
\end{cases}
\]


Thus we have reduced the problem to doubly-blinded coin-tossing with probability $\frac{f(x_i)}{n}$ (we remind that $x_1,\ldots,x_n$ are given as a ciphertexts).

\paragraph{Reducing a doubly-blinded coin-toss to $\isSmaller$.}
Given a ciphertext $\encrypted{x_i}$ and a parameter $n$, we wish to toss a coin with probability $\frac{f(x_i)}{n}$.
Uniformly draw a random value $r \leftarrow \br{0,\ldots, n}$. We observe that $Pr[r < f(x_i)] = f(x_i)/n$
and therefore is remains to test whether $r < f(x_i)$.

Since coin tossing is a basic primitive for many random algorithms, we expect our implementation of coin tossing to have a large impact on future research in HE.


\section{Protocol and Algorithms}
In this section we describe the $k$-ish nearest neighbors protocol and the algorithms it uses.
We describe it from top to down, starting with the $k$-ish nearest neighbors protocol.

\subsection{$k$-ish Nearest Neighbors}
The intuition behind our new classifier is that the classification of $k$NN does not change ``significantly'' when $k$ changes ``a little''.
See Figure~\ref{fig_kish_vs_knn} how the accuracy of $k$NN changes with $k$.
We therefore relax $k$NN into $k$-ish NN, where given a parameter $k$ the classifier considers $\kappa \approx k$.

\begin{definition}[$k$-ish Nearest Neighbors Classification]
\label{def_kish_nn}
Given a database $S=( s_1, \ldots, s_n)$, a parameter $0 \le k \le n$ and a query  $q$.
the $k$-ish Nearest Neighbors classifier sets the class of $q$ to be the majority of classes of
$\kappa$ of its nearest neighbors, where $f_1(k) < \kappa < f_2(k)$.
\end{definition}

In this paper we take $k/2 < \kappa < 3k/2$.

\subsection{$k$-ish Nearest Neighbors Classifier Protocol}
We give here a high-level description of our protocol.

The protocol has two parties a client and a server. They share common parameters: a HE scheme $\E$, the security parameter $\lambda$ and integers $p$ and $d$. 
The input of the server is a database $S=(s_1, \ldots, s_n)$ with their respective classes $class(1), \ldots, class(n)$,
where $s_i \in \ZZ_p^d$, and $class(i) \in \br{0,1}$ is the class of $s_i$.
The input of the client is query $q \in \ZZ_p^d$.

The output of the client is $class_q \in \br{0,1}$, which in high probability is the majority class of $\kappa$ neighbors of $q$, where $k/2 < \kappa < 3k/2$.
The server has no output.

Our solution starts by computing the set of distances $x_i = \dist(s_i, q)$.
Then it computes a threshold $T \assign \mu^* + \Phi^{-1}(k/n)\sigma^*$, where $\mu^* \approx E(\DD_{S,q})$ and $(\sigma^*)^2 \approx Var(\DD_{S,q})$ (see below how $\mu^*$ and $\sigma^*$ are computed).
With high probability we have $k/2 < |\br{s_i \mid x_i < T}| < 3k/2$.
Comparing $x_1,\ldots,x_n$ to $T$ is done in parallel, which keeps the depth of the circuit low.
The result of the comparison is
used to count the number of neighbors having class 0 and class 1.

To compute $\mu^*$ and $\sigma^*$ we use the identities
$\mu = 1/n \sum x_i$ and $\sigma = \sqrt{\mu^2 - 1/n \sum x_i^2}$,
and approximate $1/n \sum x_i$ and $1/n \sum x_i^2$ with an algorithm we describe in Section~\ref{sec_prob_avg}.

\paragraph{Reducing ring size.}
In the naive implementation,
we have a lower bound of $\Omega(p^2)$ for the ring size. That is because $x_1,\ldots,x_n = O(p)$ and we have the intermediate values $(\mu^*)^2, \mu_2 = O(p^2)$.
Since the size and depth of polynomial interpolations we use depend on the ring size we are motivated to keep the ring size small.

To keep the ring size small we use a representation we call {\em digit-couples representation}.
\begin{definition}[base-$p$ representation]
For $p\in\NN$ and $v\in \br{0,\ldots, p^2-1}$ base-$p$ representation of $v$ is
$\low(v) = v \mod p$ and 
$\high(v) = \floor{v/p}$.
\end{definition}
We then assign
\begin{equation*}
\begin{split}
\low(\mu_2^*) & \assign 1/n \sum x_i^2 \mod p\\
\high(\mu_2^*) & \assign \frac{1}{np} \sum x_i^2\\
\end{split}
\end{equation*}
where the modulo is done implicitly by arithmetic circuit.
Similarly, we assign
\begin{equation*}
\begin{split}
\low\big((\mu^*)^2\big) & \assign \mu^* \cdot \mu^*  \mod p,\\
\high\big((\mu^*)^2\big) & \assign \mu^* \cdot \mu^* / p,
\end{split}
\end{equation*}
where the modulo is done implicitly by arithmetic circuit.
We then assign
\[
\sigma^* =
\begin{cases}
\sqrt{\low\big((\mu^*)^2\big) - \low\big((\mu^*_2)\big)} & \text{if } \high\big((\mu^*)^2\big) - \high(\mu^*_2) = 0,\\
\sqrt{\low\big((\mu^*)^2\big) - \low\big((\mu^*_2)\big) + p} & \text{if } \high\big((\mu^*)^2\big) - \high(\mu^*_2) = 1,\\
\sqrt{\big(\high\big((\mu^*)^2\big) - \high(\mu^*_2)\big)p} & \text{otherwise.}
\end{cases}
\]
In Lemma~\ref{lemma_sigma_dagger} we prove that $\sigma^* \approx \sqrt{(\mu^*)^2 - \mu^*_2}$.

We next give a detailed description of our protocol.

\begin{protocol}
\caption{$k$-ish Nearest Neighbor Classifier Protocol\label{alg_geo_search}}
{\begin{tabbing}
	\textbf{Shared Input:} integers $p,d>1$.\\
	\textbf{Client Input:} a point $q \in \ZZ_p^d$ and a security parameter $\lambda$.\\
	\textbf{Server Input:} \= integers $k<n$, points $s_1,\ldots,s_n\in \ZZ_p^d$.\\
							\>a vector $class\in\br{0,1}^n$, s.t. $class(i)$ is the class of $s_i$.\\

	\textbf{Client Output:} \=$class_q\in \br{0,1}$, the majority class of $\kappa$ nearest neighbors of $q$\\
			 				\>where $k/2 < \kappa < 3k/2$ with high probability\\
\end{tabbing}}


{\bf Client performs:}\\
	\quad Generate keys $(sk,pk) \assign Gen(1^\lambda, p)$\label{line_key_generation}\\
	\quad $\encrypted{q} \assign Enc_{pk}(q)$\\
	\quad {\bf Send} $(pk, \encrypted{q})$ to the server\label{line_send_to_server}\\

	\BlankLine

{\bf Server performs:}\\
	\quad	\For {$\each$ $i \in 1, \ldots, n$ } {
	\quad		$\enc{x_i} \assign  \computeDist(\encrypted{q}, s_i)$\label{line_compute_distance} }

	\BlankLine

	\quad	$\encrypted{\mu^*} \assign \text{\bf approximate }  \frac{1}{n} \sum \encrypted{x_i}$ \label{line_compute_average}\\

	\BlankLine


	\quad $\big(\encrypted{\low(\mu^*_2)}, \encrypted{\high(\mu^*_2)} \big) \assign \text{ digit-couples rep. of } \frac{1}{n}\sum \encrypted{x_i}^2$\label{line_compute_second_moment}\\

	\BlankLine


	\quad $\big(\encrypted{\low\big((\mu^*)^2\big)}, \encrypted{\high\big((\mu^*)^2\big)} \big) \assign \text{ base-$p$ rep. of } (\mu^*)^2$\label{line_compute_mu_sqrt}\\

	\BlankLine

	\quad	$\encrypted{\sigma^*} \assign$ {\bf approximate} $\sqrt{(\mu^*)^2 - \mu^*_2}$\label{line_compute_sigma}\\

	\BlankLine

	\quad	$\encrypted{T^*} \assign \encrypted{\mu^*} + \round{\Phi^{-1}(\frac{k}{n})} \encrypted{\sigma^*}\label{line_compute_threshold}$\\

	\BlankLine

	\quad	$\enc{C_0} \assign \sum_{i=1}^n   \isSmaller(\encrypted{x_i},  \encrypted{(T^*)})  \cdot (1-class(i))$ \label{line_compute_class_0}\\
	\quad	$\enc{C_1} \assign \sum_{i=1}^n   \isSmaller(\encrypted{x_i},  \encrypted{(T^*)})  \cdot class(i)$ \label{line_compute_class_1} \\

	\BlankLine

	\quad	$\encrypted{class_q} \assign \isSmaller(\encrypted{C_0}, \encrypted{C_1})$\label{line_compare_classes}\\
	\quad	{\bf Send} $\encrypted{class_q}$ to the client\\

	\BlankLine

{\bf Client performs:}\\
	\quad	$class_q \assign Dec_{sk}(\encrypted{class_q})$

\end{protocol}

\paragraph{Protocol Code Explained.}
Protocol~\ref{alg_geo_search} follows the high level description given above.
%
%
As a first step the client generates a key pair $(sk,pk)$, encrypts $q$ and sends $(pk, \enc{q})$
to the Server (Line~\ref{line_key_generation}-\ref{line_send_to_server}).

The server computes the distances, $x_1, \ldots, x_n$ (Line~\ref{line_compute_distance}), where $\computeDist(s_i, q)$ computes the distance between $s_i$ and $q$.

The server then computes an approximation to the average, $\mu^*\approx 1/n \sum x_i$, (Line~\ref{line_compute_average}) by calling $\ProbAvg$ (Algorithm~\ref{alg_avg}).

In Line~\ref{line_compute_second_moment} the sever computes the base-$p$ representation of the approximation of the second moment, $\mu^*_2 = 1/n \sum x_i^2$.
\begin{equation*}
\begin{split}
\encrypted{\low(\mu^*_2)} & \assign \text{\bf approximate } \frac{1}{n} \sum \encrypted{x_i}^2 \mod p\\
\encrypted{\high(\mu^*_2)} & \assign \text{\bf approximate } \frac{1}{np} \sum \encrypted{x_i}^2 \mod p.\\
\end{split}
\end{equation*}
The approximations are done by calling $\ProbAvg$
and setting $f(x) = x^2$ and $m=n$ and $m=np$ respectively (see Section~\ref{sec_prob_avg} for details on $\ProbAvg$ parameters).
We remind that the modulo operation is performed implicitly by the HE operations.

Then  in Line~\ref{line_compute_mu_sqrt} the server computes the base-$p$ representation of $(\mu^*)^2$.
It sets:
\begin{equation*}
\begin{split}
\encrypted{\low\big((\mu^*)^2\big)} & \assign \encrypted{\mu^*} \cdot \encrypted{\mu^*}  \mod p\\
\encrypted{\high\big((\mu^*)^2\big)} & \assign \PP_{(\cdot)^2/p}(\mu^*)
\end{split}
\end{equation*}
Next the server computes $\sigma^*$ in Line~\ref{line_compute_sigma} where the square root is done with $\PP_{\sqrt{\cdot}}$ and $\PP_{\sqrt{(\cdot)+/p}}$.
In Line~\ref{line_compute_threshold} a threshold $T^*$ is computed. Since $k$ and $n$ are known $\Phi^{-1}(k/n)$ can be computed without homomorphic operations.

In Line~\ref{line_compute_class_0} the server counts the number of nearest neighbors that have class 0.
This is done by summing $\isSmaller(\encrypted{x_i},  \encrypted{(T^*)})  \cdot (1-class(i))$ for $i = 1,\ldots,n$.
Since $\isSmaller(\encrypted{x_i},  \encrypted{(T^*)}) = 1$ iff $x_i < T^*$ and $(1-class(i)) = 1$ iff $class(i)=0$ the sum in Line~\ref{line_compute_class_0} adds up
to the number of neighbors having class 0.
Similarly in Line~\ref{line_compute_class_1} the server counts the number of neighbors having class 1.

The server compares $\encrypted{C_0}$ and $\encrypted{C_1}$ in Line~\ref{line_compare_classes} to get $\encrypted{class_q}$ the majority class of the $\kappa$ nearest neighbors of $q$.
Then it sends $class_q$ to the client who can decrypt it and get the classification of $q$.


\begin{restatable}[]{theorem}{restateTheoremKNN}
Let $k,p,d\in\NN$ and $S = (s_1, \ldots, s_n) \in \ZZ_p^d$, where $class(i)$ is a class associated with $s_i$, also let $q \in \ZZ_p^d$ such that
$SD\big(\DD_{S,q}, \Gaussian(\mu, \sigma)\big) = s$, where $(\mu,\sigma)$ are the average and standard deviation of $\DD_{S, q}$.
Then:
\\
(i) The client's output in $\KNearestNeighbors$ is $class_q$ which is the majority class of $\kappa$ nearest neighbors of $q$ in $S$, where
$$Pr[\abs{k-\kappa} > \delta k] < 2\exp\Bigg(-O\bigg(\frac{\delta k(\sigma^2+\mu^2)}{\mu+\Phi^{-1}(k/n)\sigma}\bigg)\Bigg) + 2\exp\Bigg(-O\bigg(\frac{\mu\delta^2k^2\sigma}{\mu s + \sigma^2 s}\bigg)\Bigg).$$

Let $\KNearestNeighbors$ denote the arithmetic circuit evaluated by the server in Protocol~\ref{alg_geo_search} and $\isSmaller$ and $\computeDist$ denote the arithmetic circuits
comparing ciphertexts and computing the distance between two points, respectively.\\
(ii) $\depth(\KNearestNeighbors) = O(\depth(\computeDist) + \log p + \depth(\isSmaller)),$ and\\
(iii) $\size(\KNearestNeighbors) = O\big(n\cdot\size(\computeDist) + \sqrt{p} + n\cdot\size(\isSmaller)\big),$\\

where $\isSmaller$ is an arithmetic circuit comparing two ciphertexts, and $\computeDist$ is an arithmetic circuit computing the distance between two vectors.
\end{restatable}

The proof of this theorem is given in Section~\ref{sec_analysis_knn}.
In the next subsection we describe how $\mu$ and $\mu_2$ are computed efficiently in  arithmetic circuit model.

\paragraph{Increasing Success Probability.}
Since our protocol includes non-deterministic elements (the doubly-blinded coin-toss), it may fail with some probability.
The success probability can be increased by repeating the protocol multiple times and taking the majority of replies.

\paragraph{Extension to multiple database owners.}
Protocol~\ref{alg_geo_search} describes a protocol between a client and a server, however, it can be extended to a protocol where the database is distributed
among multiple data owners. The evaluation of the arithmetic circuit can then be carried collaborately or by a designated party.
In this extension the parties use multi-key HE (for example~\cite{multiKeyFHE1,multiKeyFHE2,multiKeyFHE3}) where a mutual public key is generated and
each data owner encrypt its share of the database. The protocol continues as in Protocol~\ref{alg_geo_search} to compute $\encrypted{class_q}$, which
can be decrypted by the client.

\paragraph{Extension to other distributions.}
Protocol~\ref{alg_geo_search} assumes the distribution of the distance distribution, $\DD_{S,q}$, is statistically close to Gaussian.
To extend Protocol~\ref{alg_geo_search} to another distribution $X$, the protocol needs to compute the inverse of the cumulative distribution function,
$CDF_X^{-1}(k/n)$, for any $0 < k/n < 1$.
The probability of failure will then depend on $\max CDF_X'(T)$, which intuitively bounds the change in number of nearest neighbors as $T$ changes.

\subsection{Algorithm for Computing $1/m \sum_{i=1}^n f(\encrypted{d_i})$}
\label{sec_prob_avg}
In this section we show how to efficiently approximate sums of the form $\frac{1}{m}\sum_1^n f(\encrypted{x_i})$, where $n$ and $m$ are integers, $f$ is an increasing invertible function and
$x_1,\ldots,x_n$ are ciphertexts.

\begin{algorithm}
\caption{$\ProbAvg(\encrypted{x_1}, \ldots, \encrypted{x_n})$\label{alg_avg}}
{\begin{tabbing}
	\textbf{Parameters:} \=Integers, $p,n,m > 0$, an increasing invertible function $f : \br{0,\ldots, p-1} \mapsto [0,m]$.\\
	\textbf{Input:} \=$x_1,\ldots,x_n \in \br{0,\ldots,p-1}$.\\
	\textbf{Output:} \=A number $x^* \in \br{0,\ldots,p-1}$ such that
			$Pr[\abs{\chi - x^*} > \delta] < 2e^{-2n\delta^2}$,\\
			\>where $\chi = \lceil 1/m \sum f(x_i) \rfloor \mod p$.\\
\end{tabbing}}

\vspace{-0.3cm}
	\For {$i \in 1, \ldots, n$} {
		$\encrypted{a_i} \assign $ toss a doubly-blinded coin with probability $\frac{f(x_i)}{m}$ \label{line_set_ai}\\
	}

	$\encrypted{x^*} \assign  \sum_{i=1}^{n} \encrypted{a_i} $   \label{line_add_ai}\\

	\Return $\encrypted{x^*}$
\end{algorithm}

\paragraph{Algorithm Overview.}
In Line~\ref{line_set_ai} the algorithm tosses $n$ coins with probabilities $\frac{f(x_1)}{m}, \ldots, \frac{f(x_n)}{m}$.
The coins are tossed doubly-blinded, which means the probability of each coin is a ciphertext, and the output of the toss is also a ciphertext.
See Algorithm~\ref{alg_coin_toss} to see an implementation of a doubly-blinded coin-toss.
The algorithm then returns the sum of the coin tosses, $\sum a_i$, as an estimation to $\frac{1}{m} \sum f(x_i)$.


\begin{restatable}[]{theorem}{restateProbAvg}
For any $p,m,n\in\NN$ and $f:[0,p] \mapsto [0,m]$ an increasing invertible function
Algorithm~\ref{alg_avg} describes an arithmetic circuit whose
input is $n$ integers 
$d_1,\ldots,d_n \in \br{0,\ldots,p-1}$ 
and output is $\chi^*$ such that,\\

(i) $Pr\left( \abs{\chi^* - \chi} > \delta\chi\right) < 2\exp(-\frac{\chi\delta^2}{3}),$
where $\chi = \frac{1}{m}\sum f(x)$,

(ii) $\depth(\ProbAvg) = O(\depth(\isSmaller))$,

(iii) $\size(\ProbAvg) = O(n\cdot\size(\isSmaller))$,\\
%
where $\isSmaller$ is an arithmetic circuit comparing two ciphertexts.
\end{restatable}

The proof is technical and follows from Chernoff inequality. The full proof is given in Section~\ref{sec_analysis_prob_avg}.

%
%

\subsection{Doubly Blinded Coin Toss}
\label{sec_coin_toss}

\begin{algorithm}
\caption{$\CoinToss(\encrypted{x})$\label{alg_coin_toss}}
{\begin{tabbing}
	\textbf{Parameters:} Two integers $p\in\NN$, $m\in\RR$ and an increasing invertible function $f : [0, p-1] \mapsto [0,m]$\\
	\textbf{Input:} A number $\encrypted{x}$, s.t. $x\in\br{0,\ldots,p-1}$.\\
	\textbf{Output:} A bit $\encrypted{b}$, such that $Pr[b = 1] = f(x)/m$.
\end{tabbing}}

	{\bf Draw} $r \leftarrow [0,m]$\label{line_draw}\\
	$r' \assign \ceil{f^{-1}(r)}$\label{line_inverse_f}\\

	\Return $\isSmaller(\encrypted{x}, r')$\label{line_is_smaller_const_return}
\end{algorithm}

\paragraph{Algorithm Overview.}
The $\CoinToss$ algorithm uniformly draws a random value $r$ (in plaintext) from $[0,m]$ (Line~\ref{line_draw}).
Since $r$ is not encrypted, and $f$ is increasing and invertible, it is easy to compute $\ceil{f^{-1}(r)}$ (Line~\ref{line_inverse_f}).
The algorithm then returns $\isSmaller(x, r')$ which returns 1 with probability $f(x)/m$.

\paragraph{$\CoinToss$ as an Arithmetic Circuit.}
Algorithm~\ref{alg_coin_toss} draws a number $r$ from the range $[0,m]$ and computes $f^{-1}(r)$, which are operations
that are not defined in an arithmetic circuit. To realize $\CoinToss$ as an arithmetic circuit we think of a family of circuits:
$\CoinToss_r$ for $r \in [0,m]$.
An instantiation of $\CoinToss$ is then made by drawing $r \leftarrow [0,m]$ and taking $\CoinToss_r$.

The proofs of correctness and the size and depth bounds of the arithmetic circuit implementing Algorithm~\ref{alg_coin_toss} are given in Section~\ref{sec_analysis_coin_toss}.

\section{Analysis}
In this section we prove the correctness and efficiency of our algorithms.
Unlike the algorithms that were presented top-down, reducing one problem to another simpler problem, we give the proofs bottom up as analyzing the efficiency of one algorithm builds upon the
efficiency of the simpler algorithm.

\subsection{Analyzing Doubly Blinded Coin Toss}
\label{sec_analysis_coin_toss}
In this section we prove the correctness and the bounds of the $\CoinToss$ algorithm given in Section~\ref{sec_coin_toss}.

\begin{theorem}
\label{lemma_coin_toss}
For $p\in\NN$, $m\in\RR$ and an increasing invertible function $f:[0,p-1]\mapsto[0,m]$
Algorithm~\ref{alg_coin_toss} gets an encrypted input $\encrypted{x}$ and outputs an encrypted bit $\encrypted{b}$ such that
\\
(i) $Pr[b=1] = f(x)/m$.\\
%
(ii) $\depth(\CoinToss) = O(\isSmaller)$, and\\
(iii) $\size(\CoinToss) = O(\isSmaller)$,
where $\isSmaller$ is a circuit that compares a ciphertext to a plaintext:
$\isSmaller(\encrypted{x}, y) = 1$ if $x<y$ and 0 otherwise.
\end{theorem}

\begin{proof}
{\em Correctness.}
Since $f$ is increasing and invertible
$$Pr[f(x) < r] = Pr[x < f^{-1}(r)] = Pr[x < \ceil{f^{-1}(r)}].$$
The last equation is true since since $x$ is integer.

Since we pick $r$ uniformly from $[0,m]$ we get $Pr[f(x) < r] = f(x)/m$.

{\em Depth and Size.}
After choosing $\CoinToss_r$ by randomly picking $r$, that circuit embeds $\isSmaller$ and the bound on the size and depth are immediate.

\end{proof}

The $\isSmaller$ function may be implemented differently, depending on the data representation.
In this paper, we use a polynomial interpolation to compute $\isSmaller$ and therefore,
$\depth(\isSmaller) = O(\log p)$ and $\size(\isSmaller) = O(\sqrt{p})$. We summarize it in the following corollary:

\begin{corollary}
In this paper, $\CoinToss$ is implemented with $\depth(\CoinToss) = O(\log p)$ and $\size(\CoinToss) = O(\sqrt{p})$.
\end{corollary}


\subsection{Analysis of $\ProbAvg$}
\label{sec_analysis_prob_avg}

We now prove the correctness and depth and size bounds of Algorithm~\ref{alg_avg}.

\begin{theorem}
\label{lemma_chernoff}
Let $p,m \in \NN$, $x_1, \ldots, x_n \in \br{0,\ldots, p-1}$ and $f:[0,p-1] \mapsto [0,m]$ be an increasing and invertible function.
Denote $\chi = 1/m \sum_1^n f(x_i) \mod p$ then:\\
(i) $\ProbAvg$ returns $x^*$ such that $Pr[ \abs{ x^* - \chi } > \delta\chi ] < 2\exp (-\frac{\chi\delta^2}{3})$.\\
(ii) $\depth(\ProbAvg) = O(\depth(\isSmaller))$.\\
(iii) $\size(\ProbAvg) = O(n \cdot \size(\isSmaller))$.
\end{theorem}

\begin{proof}
{\em Correctness.}
We start by proving that $\ProbAvg$ return $x^*$ such that $Pr[ \abs{ x^* - \chi } > \delta\chi ] < 2\exp (-\frac{\chi\delta^2}{3})$.
From Theorem~\ref{lemma_coin_toss} we have
\[
a_i =
\begin{cases}
1 & \text{with probability } \frac{f(x_i)}{m}\\
0 & \text{otherwise}.
\end{cases}
\]
Since $a_i$ are independent Bernoulli random variables, it follows that $E(\sum a_i) = \frac{1}{m}\sum f(x_i) = \chi$
and by Chernof we have:
$Pr\left( \sum a_i > (1+\delta)\chi\right) < \exp(-\frac{\chi\delta^2}{3})$ and
$Pr\left( \sum a_i < (1-\delta)\chi\right) < \exp(-\frac{\chi\delta^2}{2})$,
from which it immediately follows that
$Pr\left(\abs{ \sum a_i - \chi } > \delta\chi\right) < 2 \exp(-\frac{\chi\delta^2}{3})$.

{\em Depth and Size.}
We analyze the depth and size of the arithmetic circuit that implements $\ProbAvg$.
Since all coin tosses are done in parallel the multiplicative depth is $\depth(\ProbAvg) = \depth(\CoinToss)$ and the size is
$\size(\ProbAvg) = O(n\cdot\depth(\CoinToss))$.
\end{proof}


\subsection{Analysis of $\KNearestNeighbors$}
\label{sec_analysis_knn}

In this subsection we prove the correctness and bounds of the $\KNearestNeighbors$ protocol.

\restateTheoremKNN*

\begin{proof}
{\em Correctness.}
For lack of space we give the proof of correctness in Appendix~\ref{sec_proof}.
In a nutshell, the proof follows these steps:
\begin{itemize}
\item Use Theorem~\ref{lemma_chernoff} to prove $\mu^* \approx \mu$ and $\mu_2^* \approx \mu_2$ (with high probability), where
$\mu$ and $\mu_2$ are the first two moments of $\DD_{S,q}$ and $\mu^*$ and $\mu_2^*$ are the approximations calculated using $\ProbAvg$.
\item Prove $\sigma^* \approx \sigma$ (with high probability), where $\sigma = \sqrt{\mu^2 - \mu_2}$ and $\sigma^*$ is the approximation calculated by $\KNearestNeighbors$.
\item Prove $T^* \approx T$ (with high probability), where $T = \mu + \Phi^{-1}(k/n)\sigma$ and $T^* = \mu^* + \Phi^{-1}(k/n)\sigma^*$ as calculated by $\KNearestNeighbors$.
\item Prove $\abs{\sett{x_i}{x_i < T^*}} \approx \abs{\sett{x_i}{x_i < T}}$ (with high probability), where $\DD_{S,q}$ is statistically close to $\Gaussian(\mu, \sigma)$.
\end{itemize}

{\em Depth and Size.}
The protocol consists of \ref{last_step} steps:
\begin{enumerate}
\item \label{step1} Compute distances $x_1, \ldots, x_n$.
\item \label{step2} Compute $\mu^*$ and $\mu_2^*$.
\item \label{step3} Compute $(\mu^*)^2$
\item \label{step4} Compute $\sigma^*$.
\item \label{step5} Compute $T^*$.
\item \label{step6} Compute $C_0$ and $C_1$.
\item \label{step7} Compute $class_q$.
\label{last_step}
\end{enumerate}

Step~\ref{step1} is done by instantiating $n$ $\computeDist$ sub-circuits in parallel;
Step~\ref{step2} is done by instantiating $O(1)$ $\ProbAvg$ sub-circuits in parallel;
Steps~\ref{step3}-\ref{step5} are done by instantiating $O(1)$ polynomials in parallel;
Step~\ref{step6} is done by instantiating $O(n)$ $\isSmaller$ sub-circuits in parallel,
and
Step~\ref{step7} is done by instantiating $O(1)$ polynomials.

Summing it all up we get that
$$\depth(\KNearestNeighbors) = O\big(\depth(\computeDist) + \log p + \depth(\isSmaller)\big),$$ and
$$\size(\KNearestNeighbors) = O\big(n\cdot\size(\computeDist) + \sqrt{p} + n\cdot\size(\isSmaller)\big).$$

\end{proof}

Plugging in our implementations of $\isSmaller$ and $\computeDist$  we get this corollary.

\begin{corollary}
Protocol~\ref{alg_geo_search} can be implemented with
$$\depth(\KNearestNeighbors) = O(\log p),$$ and
$$\size(\KNearestNeighbors) = O(n\cdot\sqrt{p}).$$
\end{corollary}

\section{Security Analysis}
The $k$-ish Nearest Neighbors protocol involves two parties, called client and server.
The client and the server have shared parameters: a security parameter $\lambda$, a HE scheme $\E$ working over a ring of size $p$ and the dimension $d$.
In addition the client has a query given as a vector $q \in \ZZ_p^d$.
The server has additional input: two integers $k<n$, a database of vectors $S = (s_1, \ldots, s_n)$, where $s_i \in \ZZ_p^d$ and their classifications $class \in \br{0,1}^n$.
the user’s output is the class of $q$, $class_q$ determined by the majority class of $\kappa \approx k$ nearest neighbors.
The server has no output.

\begin{theorem}
The secure $k$-ish NN classifier protocol (Protocol~\ref{alg_geo_search}) securely realize the $k$-ish NN
functionality (as defined above) against a semi-honest adversary controlling the server and against
a semi-honest adversary controlling the client,  assuming the
underlying encryption $\E$ is semantically secure.
\end{theorem}

\begin{proof}
To prove the protocol is secure against a semi-honest adversary
we construct a simulator $\Sim$ whose output, when given only the server’s input and output $(1^\lambda, \E, p, d, k, n, S, class)$,
is computationally indistinguishable from an adversarial server’s view in the protocol.

The simulator operates as follows:
(i) Generates a dummy query $q'$;
(ii) Executes the $k$-ish NN classifier protocol on simulated client’s input $q'$ (the simulator plays the roles of both parties);
(iii) Outputs the sequence of messages received by the simulated server in the $k$-ish NN classifier protocol. The simulator’s output
$\Sim(\ldots) = \Sim(1^\lambda, \E, p, d, k, n, S)$ is therefore:
$$ \Sim(\ldots) = ( pk', \enc{q'}_{pk'}, \enc{class_q'}_{pk'} ),$$
where $pk'$ was generated by $\Gen(1^\lambda, p)$ and $\enc{q'}_{pk'}$ and $\enc{class_q'}_{pk'}$ were generated by $\Enc(\ldots)$.

We show that the simulator’s output is computationally indistinguishable from the view of the server
(assuming $\E$ is semantically secure).
The view of the server consists of its received messages:
$$\view(\A) = (pk, \enc{q}_{pk}, \enc{class_q}_{pk}),$$
where $pk$ was generated by $\Gen(1^\lambda, p)$ and $\enc{q}_{pk}$ and $\enc{\class_q}_{pk}$ were generated by $\Enc(\ldots)$.

Observe that the simulator’s output and the server view are identically distributed, as they are sampled
from the same distribution. Furthermore, the server view is computationally indistinguishable from the real
view by the multi-messages IND-CPA security for the HE scheme $\E$. Put together, we conclude that the
simulator’s output is computationally indistinguishable from the server’s view $\Sim(\ldots) \equiv_c \view(\A)$.

\end{proof}

\section{Experimental Results}
\label{sec_experiment}

We implemented Algorithm~\ref{alg_geo_search} in this paper and built a system that securely classifies breast  tumors.
We give the details in this section.

Our system has two parties. One party (the server)  holds a large database of breast tumors classified as malignant or benign.
The other party (the client) has a query tumor that it wishes to securely classify using the database at the server.
The client learns only the class of its tumor: benign or malignant. The server does not learn anything on the query.

We measured the time to compute the classification and the accuracy of our classifier.
In machine learning literature this is sometimes referred to as performance. We use the term accuracy since we also measure time performance.
The accuracy is expressed in terms of $F_1$ score which quantifies
the overlap between the malignant tumors and the tumor classified as malignant.

We implemented our system using HElib\cite{HELib} for HE operations and ran the server part on a standard server.
Since HElib works over an integer ring we scaled and rounded the data and the query to an integer grid.
As shown below, increasing the size of the grid improved the accuracy of the classifier but also increased the time it took to compute the classification.
As a baseline we show the accuracy of an insecure $k$NN classifier working with floating points

\subsection{The Data}
We tested our classifier with a database of breast tumors~\cite{cancerData}.
This database contains 569 tumor samples, out of which 357 were classified as benign and 212 as malignant.
Each tumor is characterized by 30 features, where each feature is a real number. Among others, these features include the
diameter of the tumor, the length of the circumference and the color it showed in an ultra-sound scan. The full list of parameters can
be found in~\cite{cancerData}.
An insecure $k$NN classifier was already suggested to classify new tumors based on this database (see for example ~\cite{kaggle}).

Although not required by our protocol we applied a preprocessing step to reduce the dimensionality of the problem.
We applied linear discriminant analysis (LDA) and projected the database onto a 2D plane.
This is a preprocessing step the server can apply on the database in clear-text before answering client queries.
Also, (even when using HE) a 30-dimensional encrypted query can easily be projected onto a lower dimensional space.
The projected database is shown in Figure~\ref{fig_projected_data}.

Since HElib encodes integer numbers from a range $\br{0,\ldots,p-1}$ for some $p$ that is determined during key generation
we scaled the data points and rounded them onto the two dimensional grid $p\times p$.
The choice of $p$ affects the accuracy as well as the running time of our algorithm.
We tested our algorithm with various grid sizes. As a baseline for accuracy we used a non-secure $k$NN classifier that works with standard floating point numbers.

\subsection{The Experiment}
\paragraph{Accuracy.}
To test the accuracy of Protocol~\ref{alg_geo_search} we used it to classify each point in the database and computed the $F_1$ score on the set of malignant tumors and the set of tumors classified as malignant.
To test a single point we removed it from the database and applied our algorithm on the smaller database with a query that was the removed point.
For a database of 568 points we used $k=13$ which sets $\Phi^{-1}(13/568) \approx 2$.
Since $\KNearestNeighbors$ is non-deterministic it considered the classes of the $\kappa$ nearest neighbors, where with high probability $k/2 < \kappa < 3k/2$.
For the classification of each point we ran 5 instances of $\KNearestNeighbors$ and took the majority.
We calculated the $F_1$ score by repeating this for each of the 569 points.
We then calculated the $F_1$ score for grids between $20\times 20$ and $300\times 300$. The results are summarized in Figure~\ref{fig_vs_gridsize}.

\paragraph{Time.}
The time to complete the $\KNearestNeighbors$ protocol comprises of 3 parts:
\begin{itemize}
\item {\bf Client Time} is the time it takes the client to process its steps of the protocol. In our case, that is (i) generating a key, (ii) encrypting a query and (iii) decrypting
the reply.
\item {\bf Communication Time} is the total time it takes to transmit messages from between the client and the server.
\item {\bf Server Time} is the time it takes the server to evaluate the arithmetic circuit of the protocol.
\end{itemize}

In our experiments,
we measured the server time, i.e. the time it took the server to evaluate the gates of the arithmetic circuit.
The time we measured was the time passed between the receiving of the encrypted query and the sending of the encrypted class of the query.
In some HE schemes the time to evaluate a single gate in a circuit depends on the depth of the entire circuit. Specifically, in the scheme we used (see below)
the overhead to compute a single gate is $\tilde{O}\big(\depth(AC)^3\big)$, where $\depth(AC)$ is the depth of the circuit.

We measured how the size of the grid affects the running time.
To do that we fixed the size of the database to 569 points and repeated the measurement for various grid sizes.
The results are summarized in Figure~\ref{fig_vs_gridsize} (left).

In addition, we measured how the size of the database affects the running time.
To do that we fixed the grid size and artificially duplicated the points of the database. In this test we were not concerned if the classification of the query
was correct or not. Since the protocol and the algorithm is oblivious to the value of the points in the database this is a good estimation to the running time of
the server for other databases as well.

\subsection{The System}
We implemented the algorithms in this paper in C++. We used HElib library~\cite{HELib} for an implementation for HE, including its
usage of SIMD (Single Instruction Multiple Data) technique,
and the LIPHE library~\cite{liphe}.
The source of our system is open under the MIT license.  

The hardware in our tests was a single off-the-shelf server with 
16 2.2 GHz Intel Xeon E5-2630 cores.
These cores are common in standard laptops and servers.
The server also had 62GB RAM, although our code used much less than that.

All the experiments we made use a security key of 80 bits. This settings is standard and can be easily changed by the client.


\subsection{Results and Discussion.}
Our results are summarized in Figures~\ref{fig_vs_gridsize} and ~\ref{fig_vs_size}.

\paragraph{$k$-ish NN vs. $k$NN.}
In Figure\ref{fig_kish_vs_knn} we show how the choice of $k$ changes the accuracy of $k$NN.
The graph shows the $F_1$ score ($y$ axis) of running $k$NN on the data with different values of $k$ ($x$ axis).
The graph shows that for $5\le k \le 20$ the decrease in $F_1$ score is miniscule, from 0.985 to 0.981.
For $20 \le k \le 375$ the $F_1$ score decreases almost linearly from 0.981 to 0.905.
For larger values, $375 < k$ the $F_1$ score drops rapidly because the $k$NN classifier considers too many neighbors.
In the extreme case, for $k = 568$ the classifier considers all data points as neighbors thus classifying all queries as benign.

The graph supports the relaxation we made to consider the $k$-ish NN.

\paragraph{Grid Size.}
The effects of the grid size are shown in Figure~\ref{fig_vs_gridsize}.

In Figure~\ref{fig_vs_gridsize} (right) we show how the accuracy, measured by $F_1$ score, changes with the grid size.
As a baseline, we compared the $F_1$ score of our protocol (shown in a solid blue line) to the $F_1$ score of a $k$NN classifier (shown in dashed red line)
also run on data rounded to the same grid size.
The $F_1$ score of the $k$NN classifier was around 0.98 and did not change much for grids of size $100\times 100$ to $200\times 200$.
The $F_1$ score of our protocol, on the other hand, increased with grid size.
For grid sizes $250\times 250$ and finer, our protocol nearly matched the accuracy of the $k$NN classifier.

For grids coarser than $250\times 250$ the performance of our classifier grew with the grid size.
For coarse grids our protocol performed worse. For example, for a $150\times 150$ grid our protocol's $F_1$ score was 0.5.
For grids between $150\times 150$ and $250\times 250$ the accuracy of our protocol improved.
Both the $k$NN and the $k$-ish NN were computed in the same grid size. The reason for the difference in their accuracy comes from the
success probability of $\ProbAvg$ (Algorithm~\ref{alg_avg}): 
$Pr\left( \abs{\mu^* - \mu} > \delta\mu\right) < 2exp(-\frac{\mu\delta^2}{3}.)$
This probability improves as the average distance $\mu$ grows, which happens when the ring size increases.
The success probability of $\ProbAvg$ then affects the success probability of $\KNearestNeighbors$ (Algorithm~\ref{alg_geo_search}).

In Figure~\ref{fig_vs_gridsize} (right) we show how the server time changed as the grid size changed.
A formula for the server time is
$Time = \tilde{O}\Big(\size(AC)\cdot \big(\depth(AC)^3\big)\Big)$, where $AC$ is the arithmetic circuit the server evaluates. Our analysis shows that for a grid $g\times g$, we have $p=O(g)$ for fixed $d$, which means
$$Time=\tilde{O}(log^3 g \sqrt{g})$$.

In Figure~\ref{fig_vs_gridsize} (left) we show how the memory requirements of evaluating the arithmetic circuit of $\KNearestNeighbors$ grows as the grid size grows.
The increase in the memory requirements comes from the polynomial interpolations we used.
Our implementation of realizing those polynomials used $O(\sqrt{g})$ memory for a $g\times g$ grid.

\paragraph{Database Size.}
The server time of our protocol changes linearly with, $n$, the size of the database, $S = (s_1, \ldots, s_n)$.
See Figure~\ref{fig_vs_size}.
This can be easily explained since the depth of the arithmetic circuit does not depend on $n$ and the number of gates needed to compute $\mu, \mu_2, C_0, C_1$ and the distances
$x_1,\ldots,x_n$ grow linearly with $n$.

\paragraph{Scaling}
Our protocol scales almost linearly with the number of cores since computing $\mu, \mu_2, C_0, C_1$ and the distances $x_1, \ldots, x_n$ are embarrassingly parallel.

\begin{figure}[!ht]
\centering
\includegraphics[width=100mm]{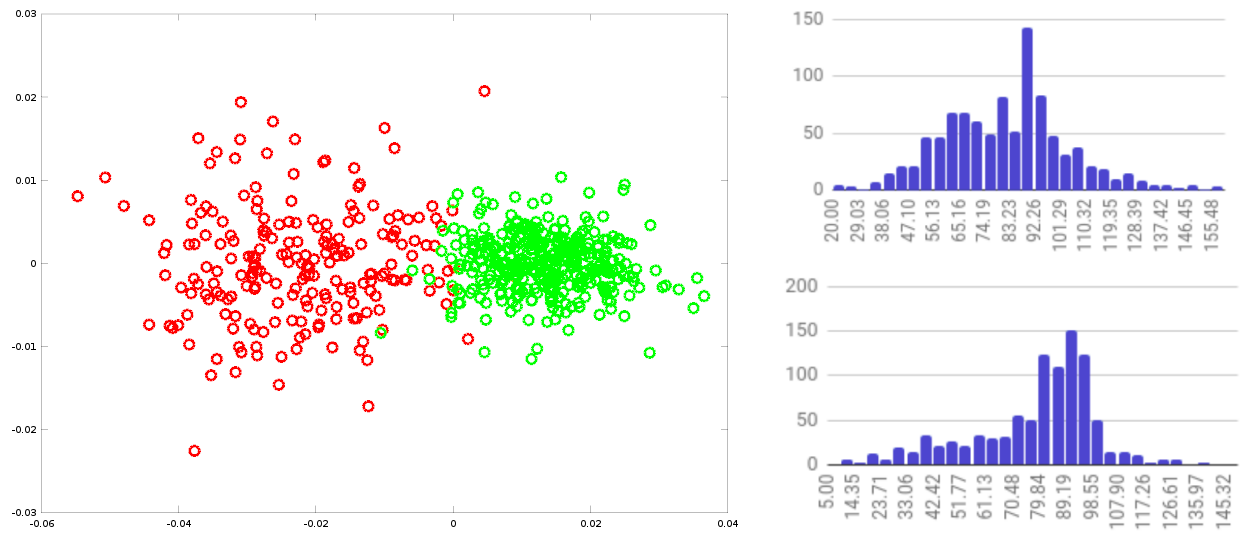}
\caption{
On the right, the  569 points in the database, representing breast tumors samples classified as benign (green) and malignant (red) after applying LDA and projecting them onto the plane.
On the left,
two histograms of distances from two random points on the plane to the 2D points in the database scaled and rounded to a $100\times 100$ grid.
}
\label{fig_projected_data}
\end{figure}


\begin{figure}[!ht]
\centering
\includegraphics[width=130mm]{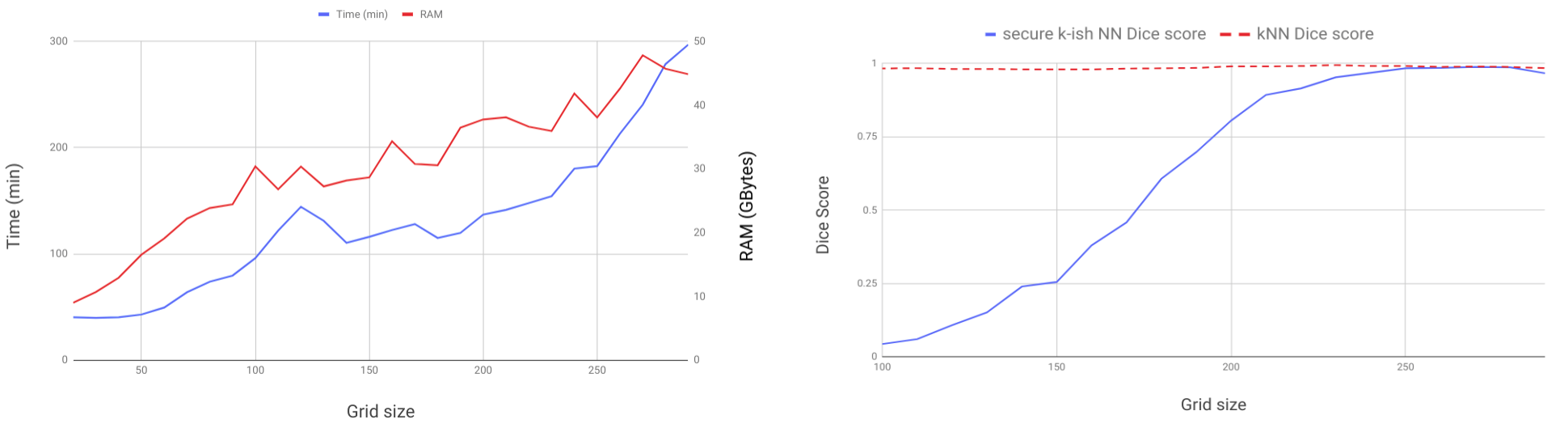}
\caption{
On the left, how the time and RAM needed to compute the $k$-ish NN grows as a function of the grid-size to which data was rounded to.
The $x$ axis is the number of grid cells in each dimension (e.g. $x=100$ means a $100\times 100$ grid).
The time (in minutes) to compute the $k$-ish NN is shown in the blue graph and it corresponds to the left vertical axis.
The RAM (in Giga Bytes) needed to compute the $k$-ish NN is shown in the red graph and it corresponds to the right vertical axis.
On the right we show in blue the $F_1$ score ($y$-axis) as a function of the grid size.
In red we show the $F_1$ score of a $k$NN ran on the database in plaintext in floating point arithmetics.
The $x$ axis is the number of grid cells in each dimension, as in the left figure.
}
\label{fig_vs_gridsize}
\end{figure}

\begin{figure}[!ht]
\centering
\includegraphics[width=130mm]{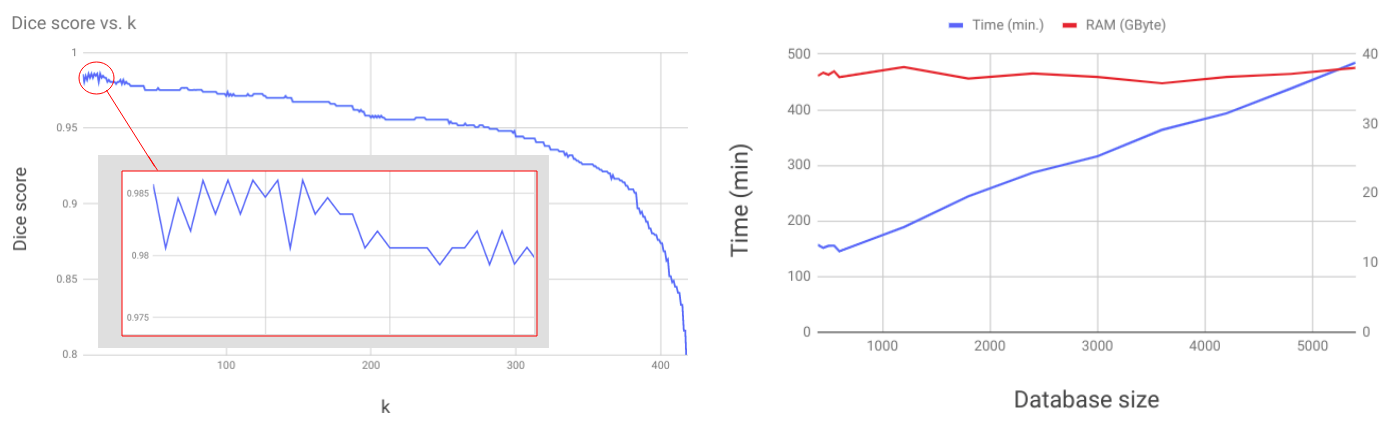}
\caption{
The left figure shows how the $F_1$ score changes as the value of $k$ changes. The $F_1$ score was calculated on a database of 569 tumors in plaintext in floating point arithmetics.
We also zoom into the range $10<x<40$, which we show under the graph.
The right figure shows how the server time (in blue) grows linearly with $n$, and how the RAM requirements (in red) do not grow as $n$ grows.
The server time correspond to the left vertical axis and the RAM requirements correspond to the right vertical axis.
}
\label{fig_kish_vs_knn}
\label{fig_vs_size}
\end{figure}

\subsection{The Naive Implementation}
In this subsection we describe the naive $k$NN implementation and its running time.
For a query $q$ and a database $S=(s_1, \ldots, s_n)$ with their classes $class(1), \ldots, class(n)$ the naive implementation follows these 3 steps:
\begin{itemize}
\item Compute distances $x_i = \computeDist(s_i, q)$ for $i=1,2\ldots,n$.
\item Sort $x_1,\ldots, x_n$ in increasing order $x_{i_1}, \ldots, x_{i_n}$.
\item Take the majority of $class(i_1), \ldots, class(i_k)$.
\end{itemize}
For the second step (sorting), we considered the work of \k{C}etin et al. \cite{Cetin15sort2,Cetin15sort1} who compared several sorting algorithms running with HE.
The fastest algorithm found by their papers was to compute $n$ polynomials $\PP_i(x_1, \ldots, x_n)$, such that $\PP_i(\ldots)$ evaluates to the $i$-th smallest value.
In their papers they gave an explicit description of $\PP_i(\ldots)$.
In our case we consider only the $k$ nearest neighbors and therefore compute only $\PP_1, \ldots, \PP_k$.

The time to evaluate $\PP_i(x_1, \ldots, x_n)$ was too high for our parameters, $n=569$.
With smaller parameters, for example, $n=10$ evaluating $\PP_1(x_1, \ldots, x_{10}), \ldots, \PP_5(x_1, \ldots, x_{10})$ took 18 hours.
Extrapolating to $n=569$ gives a running times of several weeks, which we were not able to run.

\section{Conclusions}

In this paper we present a new classifier which we call $k$-ish NN which is based on $k$NN.
In our new classifier we relax the number of neighbors used to determine the classification of a query to be approximately $k$.
We show that when the distribution of distances from a query, $q$, to the points in the database, $S$, is statistically close to Gaussian
the $k$-ish NN can be implemented as an arithmetic circuit with low depth. Specifically the depth is independent of the database size.
The depth of a circuit has a large impact when it is evaluated with HE operations. In that case, the low depth of our circuit significantly
improves the time performance.
This suggests that our solution is a good candidate to implement a secure classifier.

We give a protocol that, using the $k$-ish NN classifier, classifies a query,  $q$, based on a database $S$.
Our protocol involves two parties. One party (the client) owns $q$, while the other party (the server) owns $S$.
Our protocol involves a single round, in which the client encrypts $q$, sends it to the server, who evaluates an arithmetic circuit
whose output is the (encrypted) classification of $q$. The classification is then sent to the client, who can decrypt it.

The communication complexity of our protocol is proportional to the size of the input and output of the client and is independent of the size of $S$, unlike
previous work that are not based on HE and whose communication complexity depended on the running time which is a function of the size of $S$.
The arithmetic circuit the server evaluates in our protocol is independent of the size of $S$, unlike naive implementations with HE.
This makes our protocol the best of both worlds - having low communication complexity and low running time.

The efficiency of our protocol comes from the relaxation of the number of neighbors used in the classification of $q$.
Relaxing the number of neighbors we are able to use a non-deterministic approach to find approximately $k$ neighbors.
In the context of arithmetic circuits that means evaluating a random circuit that was drawn from a set of circuits.
In this context we have shown how to construct a doubly blinded coin toss where a coin is tossed with probability $\frac{f(\encrypted{x})}{m}$
with $\encrypted{x}$ a ciphertext, $f$ an increasing invertible function and $m > f(\encrypted{x})$.
We used doubly blinded coin tosses to efficiently approximate sums of the form $\frac{1}{m}\sum_1^n f(x_i)$ without generating large intermediate values.
Not using large intermediate values we were able to evaluate the arithmetic circuit over a small ring, which made our polynomial interpolations more efficient.

We implemented a client-server system that uses our classifier to classify real breast tumor data.
The system used HELib as a HE implementation and ran on a standard server.
Our experiments showed the $F_1$ score of our classifier was close to the $F_1$ score of plaintext $k$NN, while being significantly faster than
naive HE implementations of $k$NN.

In future work we will use our doubly blinded coin toss construction to efficiently solve more machine learning and other problems.
We believe, that a non-deterministic approach can be used to efficiently solve other problems as well, and may be of independent interest to the community.
We will also improve our $k$-ish NN to correctly find $k$-ish nearest neighbors when the distribution of distances is not statistically close to Gaussian.

%
%
%
%
%
%

\balance

\bibliographystyle{abbrv}
\bibliography{knn2}  

\appendix

\section{Proof of Correctness of $\KNearestNeighbors$}
\label{sec_proof}

In this section we prove the correctness of Protocol~\ref{alg_geo_search}.
We first prove that $\sigma^* \approx \sigma$ with high probability, where $\sigma^2 = Var(\DD_{S,q})$ and $\sigma^*$ is the approximation as computed by Protocol~\ref{alg_geo_search}.
Then we prove that $T^* \approx T$, where $T = \mu + \Phi^{-1}(k/n)\sigma$ and $T^*$ is the approximation as computed by Protocol~\ref{alg_geo_search}.
Finally, we prove that $\abs{\sett{s_i}{\dist(s_i,q) < T^*}} = \kappa \approx k$ with probability that depends on $SD(\DD_{S,q}, \Gaussian(\mu,\sigma))$.

\paragraph{Proving $\sigma^*\approx\sigma$.}

Let $\mu = \frac{1}{n}\sum x_i$ and $\mu_2 = \frac{1}{n}\sum x_i^2$ be the first two moments of $\DD_{S,q}$,
and denote $\sigma^2 = Var(\DD_{S,q}) = \mu_2 - \mu^2$.
Also denote by $\mu^* < p$ and $\mu_2^* < p^2$ the approximations of $\mu$ and $\mu_2$ as computed by Protocol~\ref{alg_geo_search} using $\ProbAvg$,
Denote $\sigma^\dagger = \sqrt{(\mu^*)^2 - \mu_2^*}$
and $\sigma^*$ be as computed by Protocol~\ref{alg_geo_search} from the digit-couples representation of $(\mu^*)^2$ and $\mu_2$,
\[
\sigma^* =
\begin{cases}
\sqrt{\big(\high((\mu^*)^2) - \high(\mu^*_2)\big)p} & \text{if } \high((\mu^*)^2) - \high(\mu_2^*) \ge 2,\\
\sqrt{p + \big(\low((\mu^*)^2) - \low(\mu^*_2)\big)} & \text{if } \high((\mu^*)^2) - \high(\mu^*_2) = 1,\\
\sqrt{\low((\mu^*)^2) - \low(\mu^*_2)} & \text{otherwise.}
\end{cases}
\]
Then,

\begin{lemma}
\label{lemma_sigma_dagger}
$$\sigma^\dagger\frac{1}{\sqrt{2}} \le \sigma^* \le \sigma^\dagger\frac{3}{\sqrt{2}}$$
\end{lemma}

\begin{proof}
When $\high((\mu^*)^2) - \high(\mu^*_2) =0$ or $\high((\mu*)^2) - \high(\mu^*_2) =1$ we have $\sigma^* = \sigma^\dagger$.
When $\high((\mu^*)^2) - \high(\mu^*_2) > 2$ we get $\sigma^\dagger > 2p$
$$\frac{1}{2} \le \frac{(\sigma^\dagger)^2 - p}{(\sigma^\dagger)^2} \le \frac{(\sigma^*)^2}{(\sigma^\dagger)^2} \le \frac{(\sigma^\dagger)^2 + p}{(\sigma^\dagger)^2} \le \frac{3}{2}$$

Therefore,
$$\frac{1}{\sqrt{2}}\sigma^\dagger \le \sigma^* \le \frac{3}{\sqrt{2}}\sigma^\dagger.$$
\end{proof}

We thus proved that $\sigma^*$ as computed by Protocol~\ref{alg_geo_search} is a good approximation for $\sigma^\dagger$. We now prove that
$\sigma^\dagger$ is a good approximation for $\sigma$ with good probability.

\begin{lemma}
Let $\mu,\mu_2, \sigma$ and $\sigma^\dagger$ be as above, then for any $\delta>0$
$$Pr[\abs{\sigma^\dagger - \sigma} < \delta\sigma] < 2e^{-\delta O(\sigma^2 +\frac{\mu^4}{\sigma^2+\mu^2} + \mu)}$$.
\end{lemma}

\begin{proof}
Set $\delta' = 2\delta + (4\delta^2 - 6\delta)\frac{\mu^2}{\mu_2}$, then
$Pr[\mu_2 > (1+\delta')\mu_2]  < \exp(-\frac{\mu_2\delta'^2}{3})$.
Since,
$(1+\delta')\mu_2 - (1-\delta)^2\mu^2 = (1+2\delta)\mu_2 - (1+2\delta)\mu^2$, we have:
$$Pr\left((\sigma^\dagger)^2 > (1+2\delta)\sigma^2\right) <\exp\left(-\frac{\mu_2\delta'^2+6\mu\delta^2}{3}\right)  $$ $$= \exp\left(-\delta O\left(\sigma^2 +\frac{\mu^4}{\mu_2} + \mu\right)\right).$$

Similarly, set $\delta'' = \delta - (\delta^2 + 3\delta)\frac{\mu^2}{\mu_2}$, then
$Pr[\mu_2 < (1-\delta'')\mu_2] < \exp(-\frac{\mu_2\delta''^2}{2})$.
Since,
$(1-\delta'')\mu_2 - (1+\delta)^2\mu^2 =
(1-\delta)\mu_2 - (1-\delta)\mu^2 $, we have:
$$Pr\left[(\sigma^\dagger)^2 < (1-\delta)\sigma^2\right] < 
\exp\left(-\frac{2\mu_2\delta''^2 + 3\mu\delta^2}{6}\right) $$ $$= \exp\left(-\delta O\left(\sigma^2 +\frac{\mu^4}{\mu_2} + \mu\right)\right).
$$

Since $\sqrt{1+2\delta} < 1+\delta$ and $\sqrt{1-\delta} < 1-\delta$ we get that
$$Pr[\abs{\sigma^\dagger - \sigma} < \delta\sigma] < \exp(-\frac{\mu_2\delta'^2+6\mu\delta^2}{3}) + \exp(-\frac{2\mu_2\delta''^2 + 3\mu\delta^2}{6}).$$

Putting it together with Lemma~\ref{lemma_sigma_dagger} we get,
$$Pr[\abs{\sigma^* - \sigma} < \delta\sigma] < 2e^{-\delta O(\sigma^2 +\frac{\mu^4}{\sigma^2+\mu^2} + \mu)}$$.

\end{proof}

\paragraph{Proving $T^*\approx T$.}

Denote by
$T = \mu + \Phi^{-1}(\frac{k}{n})\sigma$, where $\Phi$ is the CDF function of the standard Gaussian distribution, and $\Phi^{-1}$ is its inverse.
By definition, $k = \abs{ \sett{s_i}{x_i < T}}$.
Also denote by
$T^* = \mu^* + \Phi^{-1}(\frac{k}{n})\sigma^*$, as computed by Protocol~\ref{alg_geo_search}. We next show that $T^* \approx T$ with high probability.

\begin{lemma}
$$Pr[\abs{T^* - T} > \delta T] < 2\exp(-\delta O(\sigma^2 + \mu^2)) + 2\exp(\frac{-\mu\delta^2}{3}).$$
\end{lemma}

\begin{proof}
Recall that $x_1, \ldots, x_n \in \br{0, \ldots, p}$ are the distances $x_i = dist(s_i, q)$ with $\mu$ and $\sigma$ as above.
To simplify our proof we assume without loss of generality that $\Phi^{-1}(k/n) > 0$. This happens when $k>\mu$.
In cases where $k < \mu$ we replace $x_1, \ldots, x_n$ with $(p - x_1), \ldots, (p-x_n)$ and we replace $k$ with $n-k$.
By the properties of $\Phi^{-1}$ we have $\Phi^{-1}(\frac{n-k}{n}) = - \Phi^{-1}(\frac{k}{n})$.

We therefore continue assuming $\Phi(k/n)>0$.

By definition
$$\frac{T^*}{T} = \frac{\mu^*+\Phi^{-1}(k/n)\sigma^*}{\mu+\Phi^{-1}(k/n)\sigma}.$$
Since $\mu,\sigma,\Phi^{-1}(k/n) \ge 0$

$$
1-\delta = \frac{(1-\delta)\mu +(1-\delta)\Phi^{-1}(\frac{k}{n})\sigma}{\mu+\Phi^{-1}(\frac{k}{n})\sigma}<
\frac{\mu^*+\Phi^{-1}(\frac{k}{n})\sigma^*}{\mu+\Phi^{-1}(\frac{k}{n})\sigma} <
\frac{(1+\delta)\mu +(1+\delta)\Phi^{-1}(\frac{k}{n})\sigma}{\mu+\Phi^{-1}(\frac{k}{n})\sigma} = 1+\delta$$

and therefore,

$$Pr[\abs{T^* - T} > \delta T] < Pr[\abs{\mu^* - \mu} > \delta \mu] + Pr[\abs{\sigma^* - \sigma} > \delta \sigma] =
2\exp(-\delta O(\sigma^2 \mu^2)) + 2\exp(\frac{-\mu\delta^2}{3}).
$$
\end{proof}

We are now ready to prove the correctness of our protocol.

\paragraph{Proving $\kappa\approx k$.}

\restateTheoremKNN*

\begin{proof}[of item (i)]
From the definition of the cumulative distribution function (CDF) we have,
$$\frac{\kappa - k}{n} = CDF(T^*) - CDF(T).$$
Since $\DD_{S,q}$ is a discrete distribution it follows that

$$CDF(T^*) - CDF(T) = \sum_{a=0}^{T^*} Pr[x = a] -  \sum_{a=0}^{T} Pr[x = a]$$
Since 
$Pr[\abs{T^* - T} > \delta T] < 2\exp\big(-\delta O(\sigma^2 + \mu^2)\big) + 2\exp(\frac{-\mu\delta^2}{3})$ it follows with the same probability

$$\sum_{a=T(1-\delta)}^{T} Pr[x = a] <
\abs { \sum_{a=0}^{T^*} Pr[x = a] -  \sum_{a=0}^{T} Pr[x = a] }
< \sum_{a=T}^{T(1+\delta)} Pr[x = a]$$

From the definition of the statistical distance $d = \SD\big(\DD_{S,q}, \Gaussian(\mu, \sigma)\big)$ it follows that
for $n \sim \Gaussian(\mu,\sigma)$:

$$ \sum_{a=T}^{T(1+\delta)} Pr[x = a] <
\sum_{a=T}^{T(1+\delta)} (Pr[n = a] + d)  =
\sum_{a=0}^{T(1+\delta)} Pr[n = a] - \sum_{a=0}^{T} Pr[n=a] + d\delta T$$
$$ =
\Phi^{-1}(\frac{T(1+\delta) - \mu}{\sigma}) -
\Phi^{-1}(\frac{T - \mu}{\sigma})
+ d\delta T
< \frac{\sqrt{2\pi} \delta T}{\sigma} + d \delta T$$

where the last inequation is true since $(\Phi^{-1})'(x) < \frac{1}{\sqrt{2\pi}}$ and therefore $\Phi^{-1}(a+b) < \Phi^{-1}(a) + \frac{b}{\sqrt{2\pi}}$.

Similarly,
$$ \sum_{a=T(1-\delta)}^{T} Pr[x = a] <
\sum_{a=T(1-\delta)}^{T} (Pr[n = a] + d)  =
\sum_{a=0}^{T} Pr[n = a] - \sum_{a=0}^{T(1-\delta)} Pr[n=a] + d\delta T$$
$$ =
\Phi^{-1}(\frac{T - \mu}{\sigma}) -
\Phi^{-1}(\frac{T(1-\delta) - \mu}{\sigma})
+ d\delta T
< \frac{\sqrt{2\pi} \delta T}{\sigma} + d \delta T$$

where the last inequation is true since $(\Phi^{-1})'(x) < \frac{1}{\sqrt{2\pi}}$ and therefore  $\Phi^{-1}(a+b) < \Phi^{-1}(a) + \frac{b}{\sqrt{2\pi}}$.

Putting it all together we get that for any $\delta' > 0$
$$
Pr[ \abs{\kappa - k} > \delta' T (d + \frac{2\sqrt\pi}{\sigma}) ]  < 2\exp(-\delta' O(\sigma^2 + \mu^2)) + 2\exp(\frac{-\mu(\delta')^2}{3})
$$

Substituting $\delta' = \frac{\delta k}{T(d+\frac{\sqrt{2\pi}}{\sigma})}$ we get
$$
Pr[ \abs{\kappa - k} > \delta k]  < 2\exp\bigg(-O\Big(\frac{\delta k}{\mu + \Phi^{-1}(k/n)\sigma}(\sigma^2 +\frac{\mu^4}{\sigma^2+\mu^2} + \mu)\Big)\bigg) + 2\exp\bigg(-O\Big(\frac{\mu\delta^2k^2\sigma}{\mu s + \sigma^2 s}\Big)\bigg)
$$
\end{proof}

We have therefore shown that $\kappa \approx k$ with high probability.

\end{document}